\documentclass[12pt,reqno]{amsart}
\usepackage{graphicx,amscd, color,amsmath,amsfonts,amssymb,geometry,amssymb, xfrac,nicefrac}
\usepackage[initials]{amsrefs}

\newtheorem{theorem}{Theorem}[section]

\newtheorem{proposition}[theorem]{Proposition}
\newtheorem{corollary}[theorem]{Corollary}

\newtheorem{remark}[theorem]{Remark}

\newtheorem{lemma}[theorem]{Lemma}

\newtheorem{definition}[theorem]{Definition}

\numberwithin{equation}{section}

\newcommand{\Ima}{\operatorname{Im}}
\newcommand{\M}{\mathcal}

\usepackage{braket}

\begin{document}

\title{A reduction of the separability problem to SPC states in the filter normal form}

\author[Cariello ]{D. Cariello}

\address{Faculdade de Matem\'atica, \newline\indent Universidade Federal de Uberl\^{a}ndia, \newline\indent 38.400-902 Ð Uberl\^{a}ndia, Brazil.}\email{dcariello@ufu.br}


\subjclass[2010]{}

\begin{abstract}It was recently suggested  that a solution to the separability problem for states that remain positive under partial transpose composed with realignment (the so-called symmetric with positive coefficients states or simply SPC states) could  shed light on entanglement in general. Here we show that such a solution would solve the problem completely.
Given a state in $ \M{M}_k\otimes\M{M}_m$, we build a SPC state in  $ \M{M}_{k+m}\otimes\M{M}_{k+m}$ with the same Schmidt number.   It is known  that this type of state can be put in the filter normal form retaining its type. A solution to the separability problem in $\M{M}_k\otimes\M{M}_m$ could be obtained by solving the same problem for SPC states in the filter normal form within $\M{M}_{k+m}\otimes\M{M}_{k+m}$.
 This SPC state can be built arbitrarily close to the orthogonal projection on the symmetric subspace of  $ \mathbb{C}^{k+m}\otimes\mathbb{C}^{k+m}$. All the information required to understand entanglement in $ \M{M}_s\otimes\M{M}_t$ $(s+t\leq k+m)$ lies inside an arbitrarily small ball around that projection.
 We also show that the Schmidt number of any state $\gamma\in\mathcal{M}_n\otimes\mathcal{M}_n$ which commutes with the flip operator and lies inside a small ball around that projection cannot exceed $\lfloor\frac{n}{2}\rfloor$.

\end{abstract}

\vspace{0,5cm}

\maketitle

\section{Introduction}

 A discussion  on  a series of coincidences regarding a triad of quantum states was presented through the references \cites{cariello, CarielloIEEE, CarielloLMP2, CarielloArxiv}. Let the partial transpose of $\alpha\in\M{M}_k\otimes\M{M}_k$ be denoted here by $\alpha^{\Gamma}$ \cite{horodeckifamily,peres} and the realignment map be denoted by  $\mathcal{R}(\alpha)$ \cite{rudolph,rudolph2}.  The triad mentioned in these articles is formed by 
\begin{itemize}
\item the states that remain positive under partial transpose $(\alpha^{\Gamma}\geq 0)$, the so-called positive under partial transpose states or simply PPT states,
\item the states  that remain positive under partial transpose composed with realignment $(\mathcal{R}(\alpha^{\Gamma})\geq 0)$, the so-called symmetric with positive coefficients states or simply SPC states and
\item  the states that remain the same under realignment $(\alpha=\mathcal{R}(\alpha))$, the so-called invariant under realignment states.
\end{itemize}

\vspace{0,2cm}

These coincidences ultimately led to the claim that there is  a triality pattern  in entanglement theory  \cite{CarielloArxiv}, where every  proven result for one type of such states has counterparts for the other two. It was also claimed that a solution to the separability problem \cite{horodeckifamily,peres} for SPC states or invariant under realignment states could provide insights for adapting to the most important type: the positive under partial transpose states.

\vspace{0,5cm}

In this  note, we prove that SPC states are indeed extremely important to entanglement theory, since we can embed the entire set of states into SPC states and, therefore, reduce the separability problem to them. 
\vspace{0,5cm}

We begin with the notion of the Schmidt number of a bipartite mixed state $\gamma\in\mathcal{M}_k\otimes\M{M}_m$  \cite{Sperling2011, Terhal}. Consider all the ways to express the state $\gamma$ as $\sum_{i=1}^nv_iv_i^*$, where $v_i\in\mathbb{C}^k\otimes\mathbb{C}^m$, and define its Schmidt number by the following minimum over all these expressions  $$SN(\gamma)=\min\left\{\max_{i\in\{1,\ldots,n\}}\left\{SR(v_i)\right\}, \gamma=\sum_{i=1}^nv_iv_i^*\right\},$$
where $SR(v_i)$ stands for the Schmidt rank of $v_i\in\mathbb{C}^k\otimes\mathbb{C}^m$. Recall that $\gamma$ is separable if and only if $SN(\gamma)=1$.

\vspace{0,5cm}

Then we  recall another result proved in item $b)$ of \cite[Theorem 2]{CarielloLMP2}. It says that if $\beta\in \M{M}_{k}\otimes\M{M}_{k}$ is a state supported on the the anti-symmetric subspace  of $\mathbb{C}^{k}\otimes \mathbb{C}^{k}$ $($denoted here by $\mathbb{C}^{k}\wedge \mathbb{C}^{k})$ then $$SN\left(P_{sym}^{k,2}+\epsilon \frac{\beta}{tr(\beta)}\right)=\frac{1}{2}SN(\beta),$$
for every $\epsilon\in\ \left]0,\frac{1}{6}\right]$, where  $P_{sym}^{k,2}$ stands for the orthogonal projection on the symmetric subspace of $\mathbb{C}^{k}\otimes \mathbb{C}^{k}$ and $tr(\beta)$ is the trace of $\beta$.

\vspace{0,5cm}

The idea is quite simple now, given any state $\gamma\in\M{M}_k\otimes\M{M}_m$, we must create a state $\widetilde{\gamma}\in\M{M}_{k+m}\otimes\M{M}_{k+m}$ such that
\vspace{0,5cm}

 \begin{center}
$\operatorname{Im}(\widetilde{\gamma})\subset \mathbb{C}^{k+m}\wedge \mathbb{C}^{k+m}$ and $SN(\widetilde{\gamma})=2SN(\gamma)$.
\end{center}

\vspace{0,5cm}

Therefore,  for every $\epsilon\in\ \left]0,\frac{1}{6}\right]$, $$SN\left(P_{sym}^{k+m,2}+\epsilon\frac{\widetilde{\gamma}}{tr(\widetilde{\gamma})} \right)=\frac{1}{2}SN(\widetilde{\gamma})=SN(\gamma).$$

\vspace{0,5cm}
Now,
$\gamma$ is separable $(SN(\gamma)=1)$ if and only if \begin{center}$P_{sym}^{k+m,2}+\epsilon\dfrac{\widetilde{\gamma}}{tr(\widetilde{\gamma})}$ is separable $\left(SN\left(P_{sym}^{k+m,2}+\epsilon\frac{\widetilde{\gamma}}{tr(\widetilde{\gamma})} \right)=1\right)$.
\end{center} 

\vspace{0,5cm}
Next, consider the partial traces of $\gamma=\sum_{i=1}^n A_i\otimes B_i \in\mathcal{M}_k\otimes\mathcal{M}_m$ as 

\vspace{0,3cm}

\begin{center}
$tr_A(\gamma)=\sum_{i=1}^nB_itr(A_i)$\ \ \ and\ \ \ $tr_B(\gamma)=\sum_{i=1}^nA_itr(B_i)$.
\end{center}

\vspace{0,3cm}

We say that $\gamma \in\mathcal{M}_k\otimes\mathcal{M}_m$ is in the filter normal form if $tr_A(\gamma)=\frac{Id}{m}$ and $tr_B(\gamma)=\frac{Id}{k}$. In addition, $\gamma \in\mathcal{M}_k\otimes\mathcal{M}_m$ can be put in the filter normal form if there are invertible matrices $V\in\mathcal{M}_k$ and $W\in\mathcal{M}_m$ such that $(V\otimes W)\gamma (V^*\otimes W^*)$ is in the filter normal form.

\vspace{0,5cm}

This filter normal form has been used to  improve some separability criteria \cite{guhnesurvey,Git}, but it is an open problem to determine which states can be put in the filter normal form or not. Although there are partial results, such as positive definite states which can always be put in the filter normal form \cites{gurvits2004, Aubrun}. Actually, even states with small nullity can be put in this normal form \cite[Theorem 4.3]{CarielloLAMA}.

\vspace{0,5cm}
Recently, it was shown that SPC states and invariant under realignment states can also be put in the filter normal form preserving their specific structures \cite[Corollary 4.6]{CarielloArxiv}.

\vspace{0,5cm}

Here, we show that the partial transpose of  $\delta=P_{sym}^{k+m,2}+\epsilon\frac{\widetilde{\gamma}}{tr(\widetilde{\gamma})}$ is positive definite, i.e. $\delta^{\Gamma}>0$, therefore $\delta$ can be put in the filter normal form (See theorem \ref{theorem1}). Actually, in the same theorem, we obtain a stronger result, we prove that the partial transpose composed with realignment of this state is positive definite, i.e. $\mathcal{R}(\delta^{\Gamma})>0$. Hence $\delta$ is a SPC state for any state $\gamma\in\M{M}_k\otimes\M{M}_m$.

\vspace{0,5cm}

Notice that we embed the entire  set of states of $\M{M}_{k}\otimes\M{M}_{m}$ into the SPC states within a ball of radius arbitrarily small around $P_{sym}^{k+m,2}\in\M{M}_{k+m}\otimes\M{M}_{k+m}$.

\vspace{0,5cm}

As  mentioned three paragraphs above, it was proved in  \cite[Corollary 4.6]{CarielloArxiv} that there is an invertible matrix $U\in\M{M}_{k+m}$ such that $(U\otimes U)\delta(U^*\otimes U^*)$ is in the filter normal form. Now, $\mathcal{R}((U\otimes U)\delta(U^*\otimes U^*)^{\Gamma})$ remains positive definite, as explained in our corollary \ref{corollary1}. Hence $(U\otimes U)\delta(U^*\otimes U^*)$ is a SPC state in the filter normal form. Thus, we reduce the separability problem to the SPC case in the filter normal form. 

\vspace{0,5cm}

There is an existential argument that shows that the separability problem can be theoretically reduced to states in the filter normal form, however it is unsatisfactory for a couple of reasons. Here is the argument: 

\vspace{0,3cm}

\begin{quote}Given a state $\gamma\in\mathcal{M}_k\otimes\mathcal{M}_m$ and $\epsilon\in\ ]0,1[$, there are two cases: or there is a small positive $\epsilon$ such that $\delta_{\epsilon}=\epsilon (Id/k)\otimes (Id/m)+(1-\epsilon)\gamma$ is entangled and, therefore  $\gamma$ is entangled or $\delta_{\epsilon}$ is separable for every $\epsilon$ and, therefore $\gamma$ is separable. In addition $\delta_{\epsilon}$ can always be put in the filter normal form, since it is positive definite \cite[Remark 5]{Aubrun} or \cite[Theorem 4.3]{CarielloLAMA}.
\end{quote}

\vspace{0,3cm}

There are two main issues with this argument: We do not know how small $\epsilon$ must be to obtain an entangled state $\delta_{\epsilon}$ in the first case. Consequently, in order to detect the separability of $\gamma$ in the second case, we must show the separability of $\delta_{\epsilon}$ for infinitely many values of $\epsilon$. The uncertainty in the first case and the infinitely many states to be considered in the second case are very problematic. \\

In contrast, the result we present here is explicit: given a state $\gamma$, in order to prove its separability or disprove it, it is enough to do the same for $\delta=P_{sym}^{k+m,2}+\epsilon\frac{\widetilde{\gamma}}{tr(\widetilde{\gamma})}$, where $\delta$ can be put in the filter normal form. 

\vspace{0,3cm}

Moreover, our argument provides a straightforward way to produce states $\delta$ around $P^{n,2}_{sym}$ with Schmidt number varying from $1$ to $\lfloor\frac{n}{2}\rfloor$, where  $n=k+m$, $k=\lfloor\frac{n}{2}\rfloor$  and  $m=\lceil\frac{n}{2}\rceil$.
The reader may wonder if $\lfloor\frac{n}{2}\rfloor$ is an optimal upper bound for the Schmidt number of states within some small ball around $P^{n,2}_{sym}$. In our final section we prove that this is indeed the case for states that commute with the flip operator $F\in \mathcal{M}_n\otimes\mathcal{M}_n$.
\vspace{0,2cm}

This note is organized as follows: In section 2, we construct  $\widetilde{\gamma}\in\M{M}_{k+m}\otimes\M{M}_{k+m}$ described above and in section 3, we reduce the separability problem to the SPC case. In section 4, we show that the Schmidt number of any state $\gamma\in\mathcal{M}_n\otimes\mathcal{M}_n$ which commutes with the flip operator and lies inside a small ball around $P_{sym}^{n,2}$ cannot exceed $\lfloor\frac{n}{2}\rfloor$.

\vspace{0,5cm}

\section{Preliminaries}

\vspace{0,5cm}

In this section we construct  $\widetilde{\gamma}\in\M{M}_{k+m}\otimes\M{M}_{k+m}$ supported on $\mathbb{C}^{k+m}\wedge\mathbb{C}^{k+m}$ such that $SN(\widetilde{\gamma})=2SN(\gamma)$ (See lemma \ref{lemma1}), given a state $\gamma\in\M{M}_k\otimes\M{M}_m$, but first let us fix some notation.

\vspace{0,5cm}

Let $P_{anti}^{k,2}\in\M{M}_k\otimes\M{M}_k$, $P_{sym}^{k,2}\in\M{M}_k\otimes\M{M}_k$ be the orthogonal projections onto the anti-symmetric and symmetric subspaces of $\mathbb{C}^k\otimes\mathbb{C}^k$. In addition, let $V,W$ be subspaces of $\mathbb{C}^k$ and consider $V\wedge W$ as the subspace of  $\mathbb{C}^k\otimes\mathbb{C}^k$ generated by all $v\wedge w=v\otimes w - w\otimes v$, where $v\in V$ and $w\in W$.

\vspace{0,5cm}

\begin{definition}\label{definition1}Let $C=\begin{pmatrix}
Id_{k\times k}\\
0_{m\times k}
\end{pmatrix}\otimes \begin{pmatrix}
0_{k\times m}\\
Id_{m\times m}
\end{pmatrix}$ and $Q=P_{anti}^{k+m,2}C$.\end{definition}

\vspace{0,5cm}

\begin{remark}\label{remark1} Notice that $SR(Cv)=SR(v)$ and $SR(Qv)=2SR(v)$, for every $v\in\mathbb{C}^{k}\otimes\mathbb{C}^{m}$.\end{remark}

\vspace{0,5cm}

The next lemma fill in the details to construct the aforementioned $\widetilde{\gamma}\in\M{M}_{k+m}\otimes\M{M}_{k+m}$. 

\vspace{0,5cm}

\begin{lemma}\label{lemma1}Let $C,Q$ be as in definition \ref{definition1}. Then
\begin{enumerate}
\item $C^*Q=\frac{1}{2}Id\in \M{M}_{k}\otimes \M{M}_{m}$
\item $SR(C^*v)=\frac{1}{2}SR(v)$, for every $v\in (\mathbb{C}^k\times\vec{0}_{m})\wedge (\vec{0}_{k}\times\mathbb{C}^m)$.
\item $SN(Q\gamma Q^*)=2SN(\gamma)$ and $\operatorname{Im}(Q\gamma Q^*)\subset \mathbb{C}^{k+m}\wedge \mathbb{C}^{k+m}$, for every state $\gamma\in\M{M}_{k}\otimes\M{M}_{m}.$
\end{enumerate}
\end{lemma}
\begin{proof}
$1)$ Let $F\in \M{M}_{k+m}\otimes \M{M}_{k+m}$ be the flip operator and recall that $$P_{anti}^{k+m,2}=\frac{1}{2}(Id-F)\in \M{M}_{k+m}\otimes \M{M}_{k+m}.$$

Now, let $a\otimes b\in\mathbb{C}^k\otimes \mathbb{C}^m$ and consider $$C^*FC(a\otimes b)=C^*F \begin{pmatrix}
a \\
0_m
\end{pmatrix}\otimes \begin{pmatrix}
0_k \\
b
\end{pmatrix} = C^*\begin{pmatrix}
0_k \\
b
\end{pmatrix}\otimes \begin{pmatrix}
a \\
0_m
\end{pmatrix}=0_k\otimes 0_m.$$

\vspace{0.2cm}

Hence $C^*FC=0$ and $$C^*P_{anti}^{k+m,2}C=\frac{1}{2}(C^*C+C^*FC)=\frac{1}{2}C^*C=\frac{1}{2}Id\ \ \in \M{M}_{k}\otimes \M{M}_{m}.$$
\noindent $2)$ If $v\in (\mathbb{C}^k\times\vec{0}_{m})\wedge (\vec{0}_{k}\times\mathbb{C}^m)$ then there are linearly independent vectors $a_1,\ldots,a_n$ of $\mathbb{C}^k$ and linearly independent vectors $b_1,\ldots,b_n$ of $\mathbb{C}^m$ such that $$v=\sum_{i=1}^n\begin{pmatrix}
a_i \\
0_m
\end{pmatrix}\wedge \begin{pmatrix}
0_k \\
b_i
\end{pmatrix}.$$
Hence $\begin{pmatrix}
a_1 \\
0_m
\end{pmatrix},\ldots,\begin{pmatrix}
a_n \\
0_m
\end{pmatrix}, \begin{pmatrix}
0_k \\
b_1
\end{pmatrix},\ldots, \begin{pmatrix}
0_k \\
b_n
\end{pmatrix}$ are linearly independent  and $SR(v)=2n$.

\vspace{0,2cm}

Finally, notice that $C^*v=\sum_{i=1}^n a_i\otimes b_i$ andt is Schmidt rank is $n=\frac{1}{2}SR(v)$.\\

\noindent $3)$ First, by remark \ref{remark1}, $SN(Q\gamma Q^*)\leq 2SN(\gamma)$.\vspace{0,2cm}

Next, by item $1)$, $C^*Q\gamma Q^*C=\frac{1}{4}\gamma$ and, by item $2)$, $SN(C^*Q\gamma Q^*C)\leq \frac{1}{2}SN(Q\gamma Q^*)$.
\vspace{0,2cm}

These three pieces of information together imply
$$SN(\gamma)=SN(C^*Q\gamma Q^*C)\leq \frac{1}{2}SN(Q\gamma Q^*)\leq SN(\gamma).$$

Finally,  since $\operatorname{Im}(Q)\subset \mathbb{C}^{k+m}\wedge \mathbb{C}^{k+m}$, we get $$\operatorname{Im}(Q\gamma Q^*)\subset \mathbb{C}^{k+m}\wedge \mathbb{C}^{k+m}.$$
\end{proof}

\section{The Embedding and The Reduction}

The next theorem is the key to reduce the separability problem to SPC states in the filter normal form (See corollary \ref{corollary1}).

\vspace{0,5cm}

\begin{theorem}\label{theorem1}Given $\epsilon\in\ \left]0,\frac{1}{6}\right]$ and $Q$ as in definition \ref{definition1},  consider the positive map 
$T:\M{M}_k\otimes\M{M}_m\rightarrow\M{M}_{k+m}\otimes \M{M}_{k+m}$ defined by \begin{center}
$T(\gamma)=\frac{tr(\gamma)}{2}P_{sym}^{k+m,2}+\epsilon Q\gamma Q^*.$ 
\end{center}

This linear map   possesses the following properties:
 \vspace{0,2cm}

\begin{enumerate}
\item   $SN(T(\gamma))=SN(\gamma)$ for every state $\gamma\in \M{M}_k\otimes\M{M}_m$. 
\item  $T(\gamma)^{\Gamma}$ and $\mathcal{R}(T(\gamma)^{\Gamma})$ are positive definite. Hence $T(\gamma)$ is a PPT/SPC state.
\end{enumerate}
\end{theorem}\begin{proof}
First, by lemma \ref{lemma1}, $SN(Q\gamma Q^*)=2SN(\gamma)$. \vspace{0,4cm}

Notice that $$tr(Q\gamma Q^*)=tr(\gamma Q^*Q)=\frac{tr(\gamma)}{2},$$ 

since $Q^*Q=C^*Q=\frac{Id}{2}$ by item $a)$ of lemma \ref{lemma1}. Hence, by theorem \cite[Theorem 2]{CarielloLMP2},

we have  $SN(T(\gamma))=SN\left(\frac{2}{tr(\gamma)}T(\gamma)\right)$  $$\hspace{1,5cm}=SN\left(P_{sym}^{k+m,2}+\frac{\epsilon}{\ tr(Q\gamma Q^*)} Q\gamma Q^*\right)=\frac{1}{2}SN(Q\gamma Q^*)=SN(\gamma).$$

\vspace{0,4cm}

This completes the proof of item $(1)$. Let us prove item $(2)$. \vspace{0,2cm}

Recall that $P_{sym}^{k+m,2}=\frac{1}{2}(Id+F)$, where $F\in \M{M}_{k+m}\otimes  \M{M}_{k+m}$ is the flip operator. \vspace{0,2cm}

Since $F^{\Gamma}=uu^t$, where $u=\sum_{i=1}^{k+m}e_i\otimes e_i$ and $\{e_1,\ldots,e_{k+m}\}$ is the canonical basis of $\mathbb{C}^{k+m}$, $$(P_{sym}^{k+m})^{\Gamma}=\frac{1}{2}(Id+uu^t)$$

 and its smallest eigenvalue is $\frac{1}{2}$.\vspace{0,2cm}

It is known, by \cite[Lemma 3.1]{CarielloArxiv}, that

\begin{equation}\label{eq1}
\frac{\epsilon}{\ tr(Q\gamma Q^*)} \|(Q\gamma Q^*)^{\Gamma}\|_{\infty}\leq \frac{\epsilon}{\ tr(Q\gamma Q^*)} \|tr_A(Q\gamma Q^*)\|_{\infty}=\frac{\epsilon}{\ tr(tr_A(Q\gamma Q^*))} \|tr_A(Q\gamma Q^*)\|_{\infty}\leq \epsilon.
\end{equation}

\vspace{0,5cm}

Hence $$T(\gamma)^{\Gamma}=\frac{tr(\gamma )}{2}\left(\frac{1}{2}(Id+uu^t)+\frac{\epsilon}{tr(Q\gamma Q^*)} (Q\gamma Q^*)^{\Gamma}\right)$$ 

is positive definite. Now let us prove the second assertion of item $(2)$.\vspace{0,5cm}

 Since $\mathcal{R}(Id+uu^t)=uu^t+Id$, $$\mathcal{R}\left(\left(P_{sym}^{k+m}+\frac{\epsilon\ Q\gamma Q^*}{tr(Q\gamma Q^*)}\right)^{\Gamma}\right)=\frac{1}{2}(Id+uu^t)+\mathcal{R}\left(\left(\frac{\epsilon\ Q\gamma Q^*}{tr(Q\gamma Q^*)}\right)^{\Gamma}\right).$$
 
 \vspace{0,5cm}
 
 Finally, for every state $\delta$ such that $\delta F=-\delta$, we have 
 
\begin{equation}\label{eq2}
-\delta^{\Gamma}=(\delta F)^{\Gamma}=\mathcal{R}(\delta^{\Gamma}),
\end{equation} 
by item $(7)$ of \cite[Lemma 2.3]{CarielloArxiv}. \vspace{0,5cm}

 By item $3)$ of lemma \ref{lemma1}, $\operatorname{Im}(Q\gamma Q^*)\subset \mathbb{C}^{k+m}\wedge\mathbb{C}^{k+m}$, hence
 $$\frac{\epsilon\ Q\gamma Q^*}{\ tr(Q\gamma Q^*)}F=-\frac{\epsilon\ Q\gamma Q^*}{ tr(Q\gamma Q^*)}.$$
 
  Thus, by equation \eqref{eq2}, $$\mathcal{R}\left(\left(P_{sym}^{k+m}+\frac{\epsilon\ Q\gamma Q^*}{tr(Q\gamma Q^*)}\right)^{\Gamma}\right)=\frac{1}{2}(Id+uu^t)-\left(\epsilon\  \frac{Q\gamma Q^*}{tr(Q\gamma Q^*)}\right)^{\Gamma}.$$
  
  \vspace{0,2cm}

By inequality \eqref{eq1}, this matrix is positive definite.
\end{proof}

\vspace{0,5cm}

\begin{remark} Notice that $\frac{2}{tr(\gamma)}T(\gamma)$ belongs to a small ball around $P_{sym}^{k+m,2}$ for every state $\gamma\in\M{M}_k\otimes\M{M}_m$. It is an easy task to  modify our matrices $C$ and  $Q$ in order to embed every state of $ \M{M}_s\otimes\M{M}_t$ $(s+t\leq k+m)$ inside the same ball. Hence all the information required to understand entanglement in $ \M{M}_s\otimes\M{M}_t$ $(s+t\leq k+m)$ lies inside this small ball  around $P_{sym}^{k+m,2}$.
 \end{remark}

\vspace{0,2cm}

\begin{remark}Besides its Schmidt number, the spectrum of $\gamma$  is another characteristic preserved in $T(\gamma)$ as 
$\frac{2}{\epsilon}P_{anti}^{k+m,2}T(\gamma)P_{anti}^{k+m,2}$ and $\gamma$ have the same spectrum. So all these known results that relate spectrum with separability, for example \cite{nathaniel}, can still be adapted for $T(\gamma)$.
\end{remark}

\vspace{0,5cm}
 
 The next corollary is the reduction.  
 
\vspace{0,5cm}

\begin{corollary}\label{corollary1}
The separability problem in $\M{M}_k\otimes\M{M}_m$ can be reduced to states in the filter normal form, more specifically to SPC states in the filter normal form.
\end{corollary}
\begin{proof}
By the previous theorem, it is obvious that $\gamma$ is separable $(SN(\gamma)=1)$, if and only, if $T(\gamma)$ is separable $(SN(T(\gamma))=1)$. 

\vspace{0,2cm}

We also noticed there that $T(\gamma)^{\Gamma}$ is positive definite. So there are invertible matrices
$R,S\in\M{M}_{k+m}$ such that $\phi=(R\otimes S)T(\gamma)^{\Gamma}(R\otimes S)^*$ is in the filter normal form. Hence $(R\otimes \overline{S})T(\gamma)(R^*\otimes S^t)$ is in the filter normal form as well.

\vspace{0,2cm}

Actually, since  $\mathcal{R}(T(\gamma)^{\Gamma})$ is positive definite, we can do better. By  \cite[Corollary 4.6]{CarielloArxiv}, we can find an invertible matrix $O\in\M{M}_{k+m}$ such that $\delta=(O\otimes O)T(\gamma)(O^*\otimes O^*)$ is in the filter normal form. Notice that 
 $$\M{R}(\delta^{\Gamma})=(O\otimes \overline{O})\mathcal{R}(T(\gamma)^{\Gamma})(O^*\otimes O^t),$$
 by item 3 of \cite[Lemma 2.3]{CarielloArxiv}.  Therefore, $\M{R}(\delta^{\Gamma})$ is also positive definite. Hence $\delta$ is a SPC state in the filter normal form.

\vspace{0,2cm}

So we can solve the separability in $\M{M}_k\otimes\M{M}_m$ by solving the separability problem for SPC states in the filter normal form in $\M{M}_{k+m}\otimes\M{M}_{k+m}$.
\end{proof}

\vspace{0,5cm}

The reader may wonder if it is really necessary to add $tr(Q\gamma Q^*)P_{sym}^{k+m}$ to $Q\gamma Q^*$ in order to obtain a state that can be put in the filter normal form. The answer to this question is given in the next proposition, but it requires one simple definition: If $\delta=\sum_{i=1}^n A_i\otimes B_i \in \M{M}_k\otimes\M{M}_{k}$, define $G_{\delta}:\M{M}_k\rightarrow\M{M}_k$ as $G_{\delta}(X)= \sum_{i=1}^n tr(A_iX) B_i.$

\vspace{0,5cm}

\begin{proposition} If $\gamma\in \M{M}_k \otimes\M{M}_{m}$ is a state such that $k\neq m$ then $Q\gamma Q^*\in \M{M}_{k+m}\otimes \M{M}_{k+m}$ cannot be put in the filter normal form.
\end{proposition}
\begin{proof} Let $C$ and $Q$ be as in definition \ref{definition1} and define $P=P_{sym}^{k+m}C$.
Assume without loss of generality that $k>m$. 

\vspace{0,5cm}

Notice that $$\delta=\frac{1}{2}(C\gamma C^*+F(C\gamma C^*)F)=P\gamma P^*+Q\gamma Q^*,$$

where $F\in \M{M}_{k+m}\otimes  \M{M}_{k+m}$ is the flip operator.
\vspace{0,5cm}

Now, $G_{\delta}:\M{M}_{k+m}\rightarrow \M{M}_{k+m}$, as defined just before this proposition, satisfies  $$G_{\delta}\left(\begin{pmatrix}
Id_{k\times k} & 0\\
0 & 0
\end{pmatrix}\right)=\frac{1}{2}G_{C\gamma C^*}\left(\begin{pmatrix}
Id_{k\times k} & 0\\
0 & 0
\end{pmatrix}\right)+\frac{1}{2}\stackrel{=\ 0}{\overbrace{G_{FC\gamma C^*F}\left(\begin{pmatrix}
Id_{k\times k} & 0\\
0 & 0
\end{pmatrix}\right)}}=\frac{1}{2}\begin{pmatrix}
0 & 0\\
0 & tr_A(\gamma)
\end{pmatrix}.$$

\vspace{0,5cm}

Since $\delta -Q\gamma Q^*$ is positive semidefinite, the map $G_{\delta} -G_{Q\gamma Q^*}=G_{\delta -Q\gamma Q^*}$ is positive. Hence $$\operatorname{rank}\left(G_{Q\gamma Q^*}\left(\begin{pmatrix}
Id_{k\times k} & 0\\
0 & 0
\end{pmatrix}\right)\right)\leq \operatorname{rank}\left(G_{\delta}\left(\begin{pmatrix}
Id_{k\times k} & 0\\
0 & 0
\end{pmatrix}\right)\right)=\operatorname{rank}(tr_A(\gamma))\leq m<k.$$

\vspace{0,5cm}

Notice that the image of a rank $k$ positive semidefinite Hermitian matrix by  $G_{Q\gamma Q^*}$ has rank smaller than $k$. So $G_{Q\gamma Q^*}$ is not a rank non-decreasing map. The  rank non-decreasing property for $G_{Q\gamma Q^*}$ is necessary  for the state $Q\gamma Q^*$ to be put in the filter normal form \cite{CarielloLAMA, gurvits2004}. Hence $Q\gamma Q^*$ cannot be put in the filter normal form
\end{proof}

\vspace{0,5cm}

\section{The behaviour of the Schmidt number around the projection on the symmetric subspace}

\vspace{0,5cm}

Our theorem \ref{theorem1} provides a straightforward way to produce states $T(\gamma)$ around $P^{n,2}_{sym}$ with Schmidt number varying from $1$ to $\lfloor\frac{n}{2}\rfloor$  by choosing the proper $\gamma\in\mathcal{M}_{\lfloor\frac{n}{2}\rfloor}\otimes\mathcal{M}_{\lceil\frac{n}{2}\rceil}$ $(n=k+m$, $k=\lfloor\frac{n}{2}\rfloor$  and  $m=\lceil\frac{n}{2}\rceil)$.

\vspace{0,5cm}

The reader may wonder if $\lfloor\frac{n}{2}\rfloor$ is an optimal upper bound for the Schmidt number of states within some small ball around $P^{n,2}_{sym}$. In this section we prove that this is indeed the case for states that commute with the flip operator $F\in \mathcal{M}_n\otimes\mathcal{M}_n$, i.e., for states such that an orthonormal basis  of $\mathbb{C}^n\otimes \mathbb{C}^n$ formed by symmetric and antisymmetric eigenvectors exist. For states close to  $P^{n,2}_{sym}$, but not commuting with $F$, we cannot say anything.

\vspace{0,5cm}

First, we must extend a couple of results proved in \cite[Lemma 2]{CarielloLMP2} and \cite[Lemma 3]{CarielloLMP2}.
 
 \begin{lemma}\label{lemmacombinedlemmas}
 Let $a_1,a_2$ be orthonormal vectors of $\mathbb{C}^n$ and $s=a_1\otimes a_2+a_2\otimes a_1$. Then $P_{sym}^{n,2}+\pm\epsilon a_1a_1^*\otimes a_1a_1^*$ is separable for $\epsilon\in \left[0,\frac{1}{2}\right]$ and $P_{sym}^{n,2}\pm\epsilon ss^*$ is separable for $\epsilon\in \left[0,\frac{1}{12}\right]$. 
\end{lemma}
\begin{proof}
The state $P_{sym}^{n,2}-\epsilon a_1a_1^*\otimes a_1a_1^*$ was proved to be separable when $\epsilon\leq \frac{1}{2}$ in \cite[Lemma 2]{CarielloLMP2} and $P_{sym}^{n,2}-\epsilon ss^*$ was proved to be separable when  $\epsilon\in \left[0,\frac{1}{12}\right]$ in \cite[Lemma 3]{CarielloLMP2}.\vspace{0,5cm}

So it remains to show that $P_{sym}^{n,2}+\epsilon ss^*$ is separable  when $\epsilon\in \left[0,\frac{1}{12}\right]$.

Notice that $$w=a_1\otimes a_1+ a_2\otimes a_2=\frac{a_1+ia_2}{\sqrt{2}}\otimes \frac{a_1-ia_2}{\sqrt{2}}+\frac{a_1-ia_2}{\sqrt{2}}\otimes \frac{a_1+ia_2}{\sqrt{2}}$$

and $\frac{a_1+ia_2}{\sqrt{2}},\frac{a_1-ia_2}{\sqrt{2}}$ are orthonormal. So by the case already proved $P_{sym}^{n,2}-\epsilon ww^*$ is separable when $\epsilon\in \left[0,\frac{1}{12}\right]$. 

\vspace{0,3cm}

Finally notice that
\begin{equation}\label{eq3}
P_{sym}^{k,2}+\epsilon ss^*=P_{sym}^{k,2}-\epsilon ww^*+\epsilon(ww^*+ss^*).
\end{equation}
and $ww^*+ss^*=\frac{1}{2}(w+s)(w+s)^*+\frac{1}{2}(w-s)(w-s)^*$ and $w\pm s$ both have Schmidt rank equal to 1. So $ww^*+ss^*$ is separable and $P_{sym}^{k,2}+\epsilon ss^*$ is a sum of two separable states when $\epsilon\in \left[0,\frac{1}{12}\right]$ by equation \eqref{eq3}.
\end{proof}

\vspace{0,5cm}
\begin{corollary}\label{corollary2} Let $m_1,\ldots,m_{t}$, $s_1,\ldots,s_{l}$ be non-null vectors of $\mathbb{C}^ n\otimes \mathbb{C}^n$  satisfying
\begin{enumerate}
\item $m_i=v_i\otimes w_i+w_i\otimes v_i$, where $v_i$ and $w_i$ are orthogonal vectors of $\mathbb{C}^n$ 
\item $s_i=c_i\otimes c_i$, where $c_i\in \mathbb{C}^n$.
\end{enumerate}
Then $tr(\sum_{i=1}^{t}m_im_i^*)P_{sym}^{n,2}\pm \epsilon \sum_{i=1}^{t}m_im_i^*$  is separable when $\epsilon\in \left[0,\frac{1}{6}\right]$ and $tr(\sum_{i=1}^{l}s_is_i^*)P_{sym}^{n,2}\pm \epsilon\sum_{i=1}^{l}s_is_i^*$  is separable when $\epsilon\in \left[0,\frac{1}{2}\right]$.
\end{corollary}
\begin{proof}
Notice that $\frac{m_i}{\|v_i\|\|w_i\|}=\frac{v_i}{\|v_i\|}\otimes \frac{w_i}{\|w_i\|}+\frac{w_i}{\|w_i\|}\otimes \frac{v_i}{\|v_i\|}$, where $\frac{v_i}{\|v_i\|}$ and $\frac{w_i}{\|w_i\|}$ are orthonormal.\\

Therefore, by lemma \ref{lemmacombinedlemmas},$$P_{sym}^{n,2}\pm \epsilon \frac{m_i}{\|v_i\|\|w_i\|}\frac{m_i^*}{\|v_i\|\|w_i\|}=P_{sym}^{n,2}\pm \epsilon 2\frac{m_im_i^*}{tr(m_im_i^*)}$$

 is separable when $\epsilon\in \left[0,\frac{1}{12}\right]$.
Thus, $$tr(m_im_i^*)P_{sym}^{n,2}\pm \epsilon m_im_i^*$$

 is separable when $\epsilon\in \left[0,\frac{1}{6}\right]$.\\

The other case is similar.  Notice that $\frac{s_i}{\|c_i\|^ 2}= \frac{c_i}{\|c_i\|}\otimes \frac{c_i}{\|c_i\|}$. Therefore, by lemma \ref{lemmacombinedlemmas}, 
$$P_{sym}^{n,2}\pm \epsilon \frac{s_i}{\|c_i\|^ 2}\frac{s_i^*}{\|c_i\|^ 2}=P_{sym}^{n,2}\pm \epsilon \frac{s_is_i^*}{tr(s_is_i^*)}$$

 is separable when $\epsilon\in \left[0,\frac{1}{2}\right]$.Thus, $$tr(s_is_i^*)P_{sym}^{n,2}\pm \epsilon s_is_i^*$$

 is separable when $\epsilon\in \left[0,\frac{1}{2}\right]$.\\
\end{proof}

\vspace{0,5cm}

\begin{lemma}\label{lemmadetail}Let $v=\sum_{i=1}^t\lambda_i a_i\otimes a_i$ where $\lambda_i>0$, for every $i$, and $a_1,\ldots,a_t$ are orthonormal vectors of $\mathbb{C}^n$. There are $m_1,\ldots,m_{2^{t-1}-1}$, $s_1,\ldots,s_{2^{t-1}}$ non-null vectors of $\mathbb{C}^ n\otimes \mathbb{C}^n$  satisfying
\begin{enumerate}
\item $m_i=v_i\otimes w_i+w_i\otimes v_i$, where $v_i$ and $w_i$ are orthogonal vectors of $\mathbb{C}^n$ 
\item $s_i=c_i\otimes c_i$, where $c_i\in \mathbb{C}^n$
\item $A=vv^*+\sum_{i=1}^{2^{t-1}-1}m_im_i^*=\sum_{i=1}^{2^{t-1}}s_is_i^*$
\item $tr(A)=(\lambda_1+\ldots+\lambda_t)^2$
\end{enumerate}
 \end{lemma}
 \begin{proof}
 It is an induction on $t$. 
 
 If $t=2$ then $v=\lambda_1 a_1\otimes a_1+\lambda_2 a_2\otimes a_2$. Set $m_1=\sqrt{\lambda_1\lambda_2} (a_1\otimes a_2+a_2\otimes a_1)$.
 
 Notice that $$vv^*+m_1m_1^*=\left(\frac{v+m_1}{\sqrt{2}}\right)\left(\frac{v+m_1}{\sqrt{2}}\right)^*+\left(\frac{v-m_1}{\sqrt{2}}\right)\left(\frac{v-m_1}{\sqrt{2}}\right)^*.$$

Since $$s_1= \frac{v+m_1}{\sqrt{2}}=\frac{\lambda_1+\lambda_2}{\sqrt{2}}\left(\frac{\sqrt{\lambda_1}a_1+\sqrt{\lambda_2}a_2}{\sqrt{\lambda_1+\lambda_2}}\right)\otimes\left(\frac{\sqrt{\lambda_1}a_1+\sqrt{\lambda_2}a_2}{\sqrt{\lambda_1+\lambda_2}}\right)\text{ and}$$
$$s_2= \frac{v-m_1}{\sqrt{2}}=\frac{\lambda_1+\lambda_2}{\sqrt{2}}\left(\frac{\sqrt{\lambda_1}a_1-\sqrt{\lambda_2}a_2}{\sqrt{\lambda_1+\lambda_2}}\right)\otimes\left(\frac{\sqrt{\lambda_1}a_1-\sqrt{\lambda_2}a_2}{\sqrt{\lambda_1+\lambda_2}}\right), $$

$tr(s_is_i^*)=\frac{(\lambda_1+\lambda_2)^2}{2}$ for $i=1,2$. Therefore $$tr(vv^*+m_1m_1^*)=tr(s_1s_1^*+s_2s_2^*)=(\lambda_1+\lambda_2)^2.$$

The proof of the case $t=2$ is complete. Suppose the result holds for $t=l$. \\

Let $t=l+1$. So $v=\sum_{i=1}^{l+1}\lambda_i a_i\otimes a_i$. Set $w=\sqrt{\lambda_1\lambda_2} (a_1\otimes a_2+a_2\otimes a_1)$. Again notice that 
$$vv^*+ww^*=\left(\frac{v+w}{\sqrt{2}}\right)\left(\frac{v+w}{\sqrt{2}}\right)^*+\left(\frac{v-w}{\sqrt{2}}\right)\left(\frac{v-w}{\sqrt{2}}\right)^*,$$

but now $$\frac{v\pm w}{\sqrt{2}}=\frac{\lambda_1+\lambda_2}{\sqrt{2}}\left(\frac{\sqrt{\lambda_1}a_1\pm\sqrt{\lambda_2}a_2}{\sqrt{\lambda_1+\lambda_2}}\right)\otimes\left(\frac{\sqrt{\lambda_1}a_1\pm\sqrt{\lambda_2}a_2}{\sqrt{\lambda_1+\lambda_2}}\right)+\sum_{i=3}^{l+1}\frac{\lambda_i}{\sqrt{2}}a_i\otimes a_i.$$

By induction hypothesis, there are $m_1^{\pm},\ldots,m_{2^{l-1}-1}^{\pm}$, $s_1^{\pm},\ldots,s_{2^{l-1}}^{\pm}$ non-null vectors of $\mathbb{C}^ n\otimes \mathbb{C}^n$  satisfying the first two hypothesis and  $$A_{\pm}=\left(\frac{v\pm w}{\sqrt{2}}\right)\left(\frac{v\pm w}{\sqrt{2}}\right)^*+\sum_{i=1}^{2^{l-1}-1}(m_i^{\pm})(m_i^{\pm})^*=\sum_{i=1}^{2^{l-1}}(s_i^{\pm})(s_i^{\pm})^*.$$
In addition $tr(A_{\pm})=\left(\frac{\lambda_1+\lambda_2}{\sqrt{2}}+\sum_{i=3}^ {l+1}\frac{\lambda_i}{\sqrt{2}}\right)^2=\frac{(\lambda_1+\ldots+\lambda_{l+1})^ 2}{2}$.\\

So $tr(A_++A_{-})=(\lambda_1\ldots+\lambda_{l+1})^2$ and $A_++A_{-}=$
$$vv^*+\stackrel{2^l-1 \text{ terms }}{\overbrace{ww^*+\sum_{i=1}^{2^{l-1}-1}(m_i^{+})(m_i^{+})^*+\sum_{i=1}^{2^{l-1}-1}(m_i^{-})(m_i^{-})^*}}=\stackrel{2^l \text{ terms }}{\overbrace{\sum_{i=1}^{2^{l-1}}(s_i^{+})(s_i^{+})^*+\sum_{i=1}^{2^{l-1}}(s_i^{-})(s_i^{-})^*}}.$$
 \end{proof}

\begin{lemma}\label{lemmakeyseparable}Let $v$ be a unit symmetric vector of $\mathbb{C}^n\otimes\mathbb{C}^n$ with Schmidt rank equal to $t$. If $\epsilon\in\left[0,\frac{1}{6(2t-1)}\right]$ then $P_{sym}^{n,2}\pm\epsilon vv^*$ is separable. \end{lemma}
\begin{proof} Let $\epsilon=\epsilon_1\epsilon_2$, where $\epsilon_1=\frac{1}{6}$ and $\epsilon_{2}\leq \frac{1}{2t-1}$. Then $$P_{sym}^{n,2}\pm\epsilon vv^*=P_{sym}^{n,2}\pm\epsilon_1 v_1v_1^*,$$ 

where $v_1=\sqrt{\epsilon_2}v$.\\

By item $c)$ of \cite[Corollary 4.4.4]{Horn}, there are orthonormal vectors $a_1,\ldots,a_t$ of $\mathbb{C}^n$ and positive numbers $\lambda_1,\ldots,\lambda_t$ such that $v_1=\sum_{i=1}^t\lambda_ia_i\otimes a_i$. Notice that $\lambda_1^ 2+\ldots+\lambda_t^2=\|v_1\|^2=\epsilon_2\|v\|^2=\epsilon_2$.  

Now, by lemma \ref{lemmadetail}, for this $v_1$, there are  $m_1,\ldots,m_{2^{t-1}-1}$, $s_1,\ldots,s_{2^{t-1}}$ non-null vectors of $\mathbb{C}^ n\otimes \mathbb{C}^n$  satisfying the hypothesis of that lemma. Hence $A=\sum_{i=1}^{2^{t-1}-1}m_im_i^*+v_1v_1^*=\sum_{i=1}^{2^{t-1}}s_is_i^*$ and\\

$$P_{sym}^{n,2}\pm\epsilon_1 v_1v_1^*=P_{sym}^{n,2}\mp\epsilon_1\left(\sum_{i=1}^{2^{t-1}-1}m_im_i^*\right) \pm\epsilon_1 \left(\sum_{i=1}^{2^{t-1}-1}m_im_i^*+v_1v_1^*\right)$$
\begin{equation}\label{eq4}
\hspace{0,8cm} =P_{sym}^{n,2}\mp\epsilon_1\left(\sum_{i=1}^{2^{t-1}-1}m_im_i^*\right) \pm\epsilon_1 \left(\sum_{i=1}^{2^{t-1}}s_is_i^*\right).
\end{equation}

\vspace{0,5cm}

Rewriting equation \eqref{eq4} we get $P_{sym}^{n,2}\pm\epsilon_1 v_1v_1^*=$
\begin{equation}\label{eq5}
=\left(1-tr\left(\sum_{i=1}^{2^{t-1}-1}m_im_i^*\right)-tr\left(\sum_{i=1}^{2^{t-1}}s_is_i^*\right)\right)P_{sym}^{n,2}+
\end{equation}
\begin{equation}\label{eq6}
tr\left(\sum_{i=1}^{2^{t-1}-1}m_im_i^*\right)P_{sym}^{n,2}\mp\epsilon_1\left(\sum_{i=1}^{2^{t-1}-1}m_im_i^*\right)+
\end{equation}
\begin{equation}\label{eq7}
tr\left(\sum_{i=1}^{2^{t-1}}s_is_i^*\right)P_{sym}^{n,2} \pm\epsilon_1 \left(\sum_{i=1}^{2^{t-1}}s_is_i^*\right).
\end{equation}

\vspace{0,5cm}

Notice the matrices in \eqref{eq6} and \eqref{eq7} are separable, since $\epsilon_1=\frac{1}{6}$ and by lemma \ref{corollary2}.

\vspace{0,5cm}

Next, by the hypothesis of lemma \ref{lemmadetail},  $tr(A)=tr\left(\sum_{i=1}^{2^{t-1}}s_is_i^*\right)=(\lambda_1+\ldots+\lambda_t)^2$, so
$$tr\left(\sum_{i=1}^{2^{t-1}-1}m_im_i^*\right)=(\lambda_1+\ldots+\lambda_t)^2-tr(v_1v_1^*)=(\lambda_1+\ldots+\lambda_t)^2-\lambda_1^2-\ldots-\lambda_t^2.$$

\vspace{0,5cm}

Hence 
$$\left(1-tr\left(\sum_{i=1}^{2^{t-1}-1}m_im_i^*\right)-tr\left(\sum_{i=1}^{2^{t-1}}s_is_i^*\right)\right)=1+\lambda_1^2+\ldots+\lambda_t^2 -2(\lambda_1+\ldots+\lambda_t)^2.$$

\vspace{0,5cm}

Since $(\lambda_1+\ldots+\lambda_t)^2\leq (\lambda_1^2+\ldots+\lambda_t^2)t$ and $\lambda_1^2+\ldots+\lambda_t^2=\epsilon_2\leq\frac{1}{2t-1}$, 

$$\left(1-tr\left(\sum_{i=1}^{2^{t-1}-1}m_im_i^*\right)-tr\left(\sum_{i=1}^{2^{t-1}}s_is_i^*\right)\right)\geq 1-(2t-1)(\lambda_1^2+\ldots+\lambda_t^2)\geq 0.$$

 So the matrix in \eqref{eq5} is separable, which completes the proof.
\end{proof}

\vspace{0,5cm}

In order to describe the next results, we must use the trace norm of a matrix $P$, denoted here by $\|P\|_1$. This norm is defined as the sum of its singular values.

\vspace{0,5cm}

\begin{corollary}\label{corollary3}
Let $P\in \mathcal{M}_n\otimes\mathcal{M}_n$ be Hermitian and $\Ima(P)$ be a subspace of the symmetric subspace of $\mathbb{C}^n\otimes\mathbb{C}^n$. If $\epsilon\in\left[0,\frac{1}{6(2n-1)}\right]$ and $\|P\|_1=1$ then $P_{sym}^{n,2}\pm \epsilon P$ is separable.
\end{corollary}
\begin{proof}
Let $P=P_1-P_2$, where each $P_i$ is positive semidefinite, its image is a subspace of the symmetric subspace of $\mathbb{C}^n\otimes\mathbb{C}^n$ and $\|P\|_1=\|P_1\|_1+\|P_2\|_1$.

\vspace{0,5cm}

Consider the spectral decompositions of $P_1$ and $P_2$:
\begin{center}
$P_1=\sum_{i=1}^s\lambda_iv_iv_i^*$ and $P_2=\sum_{j=1}^l m_jr_jr_j^*$,
\end{center}
where $\lambda_i>0$ and $m_j>0$, for every $i,j$, and $v_1,\ldots,v_s$ are orthonormal symmetric vectors and the same is valid for  $r_1,\ldots,r_l$. Thus, $\|P\|_1=\sum_{i=1}^s\lambda_i+\sum_{j=1}^l m_j=1$.

\vspace{0,5cm}

Therefore, $P_{sym}^{n,2}\pm \epsilon P=P_{sym}^{n,2}\pm \epsilon (P_1-P_2)=$

\begin{equation}
\label{eq8}
=\sum_{i=1}^s \lambda_i(P_{sym}^{n,2}\pm \epsilon v_iv_i^*)+\sum_{j=1}^l m_j(P_{sym}^{n,2}\mp \epsilon r_jr_j^*).
\end{equation}

Finally, since the Schmidt rank of each $v_i$ and $r_j$ is less or equal to $n$ and $\epsilon\in\left[0,\frac{1}{6(2n-1)}\right]$, each matrix inside parenthesis of the equation \eqref{eq8} is separable by Lemma \ref{lemmakeyseparable}.
\end{proof}

\vspace{0,5cm}

We can finally prove the main result of this section. 

\vspace{0,5cm}

\begin{theorem}Let $\gamma\in\mathcal{M}_n\otimes \mathcal{M}_n$ be a state that commutes with the flip operator $F\in\mathcal{M}_n\otimes \mathcal{M}_n$. If $\left\|P_{sym}^{n,2}-\gamma\right\|_1\leq \frac{1}{6(2n-1)}$ then $SN(\gamma)\leq \lfloor\frac{n}{2}\rfloor$.
\end{theorem}
\begin{proof}
We can write $\gamma=P_{sym}^{n,2}+\epsilon P$, where $\|P\|_1=1$ and $\epsilon\in \left[0,\frac{1}{6(2n-1)}\right]$.

\vspace{0,5cm}

Since $\gamma$ is a state and both $\gamma$ and $P_{sym}^{n,2}$ commute with flip operator $F\in\mathcal{M}_n\otimes \mathcal{M}_n$, $P$ is Hermitian and commutes with $F$.
Hence $P=A+B$, where $A$ is Hermitian and its image is a subspace of the symmetric subspace of $\mathbb{C}^n\otimes\mathbb{C}^n$ and $B$ is Hermitian and its image is a subspace of the anti-symmetric subspace of $\mathbb{C}^n\otimes\mathbb{C}^n$. In addition, $1=\|P\|_1=\|A\|_1+\|B\|_1$ and, without loss of generality, we can assume that  $\|A\|_1,\|B\|_1$ are not zero.

\vspace{0,5cm}

As  $\gamma=P_{sym}^{n,2}+\epsilon A+\epsilon B$ is positive semidefinite and the images of $A$ and $P_{sym}^{n,2}$ are inside the symmetric subspace of $\mathbb{C}^n\otimes\mathbb{C}^n$, the matrix $B$ is positive semidefinite.

\vspace{0,5cm}

Now, $$\gamma=\|A\|_1\left(P_{sym}^{n,2}+\epsilon \frac{A}{\|A\|_1}\right)+\|B\|_1\left(P_{sym}^{n,2}+\epsilon \frac{B}{\|B\|_1}\right).$$ 

\vspace{0,5cm}

Finally,  $SN\left(P_{sym}^{n,2}+\epsilon \frac{A}{\|A\|_1}\right)=1$ by corollary \ref{corollary3} and $SN\left(P_{sym}^{n,2}+\epsilon \frac{B}{\|B\|_1}\right)=\frac{SN(B)}{2}\leq\lfloor\frac{n}{2}\rfloor$ by  \cite[Theorem 2]{CarielloLMP2}. So the equation above implies $SN(\gamma)\leq\lfloor\frac{n}{2}\rfloor$.
\end{proof}

\vspace{0,5cm}

\section*{Summary and Conclusion}
In this work we reduced the separability problem to SPC states in the filter normal form. Given a state in $\M{M}_k\otimes\M{M}_m$, we built a SPC state in $\M{M}_{k+m}\otimes\M{M}_{k+m}$ with the same Schmidt number without any knowledge of its value. Therefore, in order to solve the separability problem in $\M{M}_k\otimes\M{M}_m$ for any kind of state, we can do it by solving the problem for SPC states in $\M{M}_{k+m}\otimes\M{M}_{k+m}$ . It is known that such states can be put in the filter normal form preserving its SPC structure. We also showed that the Schmidt number of any state $\gamma\in\mathcal{M}_n\otimes\mathcal{M}_n$ which commutes with the flip operator and lies inside a small ball around $P_{sym}^{n,2}$ cannot exceed $\lfloor\frac{n}{2}\rfloor$.

\vspace{0,5cm}

\section*{Acknowledgment} The author would like to thank the referee for providing constructive comments and helping in the improvement of this manuscript. 

\section*{Disclosure Statement}
\vspace{0,2cm}
No potential conflict of interest was reported by the author.
\vspace{0,5cm}
\section*{Data Availability Statement}
\vspace{0,2cm}
Data sharing not applicable to this article as no datasets were generated or analysed during the current study.
\vspace{0,2cm}

\end{document}